\documentclass[a4paper,UKenglish,cleveref, autoref, thm-restate]{lipics-v2021}

\nolinenumbers
\title{(Independent) Roman Domination Parameterized by Distance to Cluster} 

\authorrunning{ }

   \author{Pradeesha Ashok }{International Institute of Information Technology Bangalore, India}{International Institute of Information Technology Bangalore, India}{https://orcid.org/0000-0002-1825-0097}{}

    \author{Gautam K. Das}{Department of Mathematics, Indian Institute of Technology Guwahati, India}{gkd@iitg.ac.in}{}{}
    
    \author{Arti Pandey}{Department of Mathematics, Indian Institute of Technology Ropar, India}{arti@iitrpr.ac.in}{}{}
    
    \author{Kaustav Paul}{Department of Mathematics, Indian Institute of Technology Ropar, India}{kaustav.20maz0010@iitrpr.ac.in}{}{}

    \author{Subhabrata Paul}{Department of Mathematics, Indian Institute of Technology Patna, India}{subhabrata@iitp.ac.in}{}{}

\keywords{Roman Domination, Independent Roman Domination, FPT, Distance to Cluster, Kernel}

\Copyright{ }
\ccsdesc[100]{{Design and Analysis of Algorithms}}

\begin{document}

\maketitle

\begin{abstract}
Given a graph $G=(V,E)$, a function $f:V\to \{0,1,2\}$ is said to be a \emph{Roman Dominating function} (RDF) if for every $v\in V$ with $f(v)=0$, there exists a vertex $u\in N(v)$ such that $f(u)=2$. A Roman Dominating function $f$ is said to be an \emph{Independent Roman Dominating function} (IRDF), if $V_1\cup V_2$ forms an independent set, where $V_i=\{v\in V~\vert~f(v)=i\}$, for $i\in \{0,1,2\}$. The total weight of $f$ is equal to $\sum_{v\in V} f(v)$, and is denoted as $w(f)$. The \emph{Roman Domination Number} (resp. \emph{Independent Roman Domination Number}) of $G$, denoted by $\gamma_R(G)$ (resp. $i_R(G)$), is defined as min$\{w(f)~\vert~f$ is an RDF (resp. IRDF) of $G\}$. For a given graph $G$, the problem of computing $\gamma_R(G)$ (resp. $i_R(G)$) is defined as the \emph{Roman Domination problem} (resp. \emph{Independent Roman Domination problem}).

In this paper, we examine structural parameterizations of the (Independent) Roman Domination problem. We propose fixed-parameter tractable (FPT) algorithms for the (Independent) Roman Domination problem in graphs that are $k$ vertices away from a cluster graph. These graphs have a set of $k$ vertices whose removal results in a cluster graph. We refer to $k$ as the distance to the cluster graph. Specifically, we prove the following results when parameterized by the deletion distance $k$ to cluster graphs: we can find the Roman Domination Number (and Independent Roman Domination Number) in time $4^kn^{O(1)}$. In terms of lower bounds, we show that the Roman Domination number can not be computed in time $2^{\epsilon k}n^{O(1)}$, for any $0<\epsilon <1$ unless a well-known conjecture, SETH fails. In addition, we also show that the  Roman Domination problem parameterized by distance to cluster, does not admit a polynomial kernel unless NP $\subseteq$ coNP$/$poly.

\end{abstract}

\section{Introduction}
The concept of Roman Dominating function originated in an article by Ian Stewart, titled ``Defend the Roman Empire!$"$ \cite{stewart1999defend}, published in Scientific American. Given a graph, where every vertex corresponds to a distinct geographical region within the historical narrative of the Roman Empire, the characterization of a location as secured or unsecured is delineated by the Roman Dominating function, denoted as $f$.

Specifically, a vertex $v$ is said to be unsecured if it lacks stationed legions, expressed as $f(v) = 0$. Conversely, a secured location is one where one or two legions are stationed, denoted by $f(v) \in \{1, 2\}$. The strategic methodology for securing an unsecured area involves the deployment of a legion from a neighboring location.

In the fourth century A.D., Emperor Constantine the Great enacted an edict precluding the transfer of a legion from a fortified position to an unfortified one if such an action would result in leaving the latter unsecured. Therefore, it is necessary to first have two legions at a given location ($f(v) = 2$) before sending one legion to a neighbouring location. This strategic approach, pioneered by Emperor Constantine the Great, effectively fortified the Roman Empire. Considering the substantial costs associated with legion deployment in specific areas, the Emperor aimed to strategically minimize the number of legions required to safeguard the Roman Empire.

The notion of Roman Domination in graphs was first introduced in an article by Ian Stewart \cite{stewart1999defend}. Given a graph $G=(V,E)$, a \emph{Roman Dominating function} (RDF) is defined as a function  $f: V \to \{0,1,2\}$, where every vertex $v$, for which $f(v)=0$ must be adjacent to at least one vertex $u$ with $f(u)=2$. The weight of an RDF is defined as $w(f)=\sum_{v\in V} f(v)$. The \emph{Roman Domination Number} is defined as $\gamma_R(G)= \min\{w(f)~\vert~f$ is an RDF of $G\}$. While the context is clear, if $f(v)=i$ for some RDF $f$, then we say that $v$ has label $i$.

 Given a graph $G=(V,E)$, a set $S\subseteq V$ is defined as \emph{independent set} if any two vertices of $S$ are non-adjacent. A function $f$ is referred to as an \emph{Independent Roman Dominating function} (\emph{IRDF}) if $f$ is an RDF and $V_1\cup V_2$ is an independent set.  The \emph{Independent Roman Domination Number} is defined as $i_R(G)= \min\{w(f)~\vert~f$ is an IRDF of $G\}$. An IRDF $f$ of $G$ with $w(f)=i_R(G)$ is denoted as an $i_R(G)$-function of $G$. Given a graph $G=(V,E)$, the problem of computing $i_R(G)$ is known as \emph{Independent Roman Domination problem}.
 
 One of the objectives of parameterized complexity is to identify parameters that render NP-hard problems fixed-parameter tractable (FPT). This is of practical significance because there are often small parameters, aside from solution size, that capture important practical inputs. Hence, it only makes sense to explore problems under a multitude of parameters. There has recently been a lot of research in this area. A key research direction involves identifying a parameter as small as possible, under which a problem becomes fixed-parameter tractable or admits a polynomial-sized kernel. Structural parameterization involves a parameter that is a function of the input structure rather than the standard output size. A recent trend in structural parameterization is to study problems parameterized by the deletion distance to various graph classes where the problem is efficiently solvable.

 Our parameter of interest is the `distance' of the graph from a natural class of graphs. Here, `distance' refers to the number of vertices that must be deleted from the graph to belong to the specified class. This article focuses on one such special class of graphs: cluster graphs, where each connected component of the graph is a clique. Note that both Roman Domination and Independent Roman Domination problems can be easily solved in cluster graphs. Given a graph $G=(V,E)$, the minimum number of vertices that need to be deleted from the graph so that the remaining graph becomes a cluster graph is called \emph{distance to cluster} of $G$ (denoted by CVD size of $G$).

\section{Preliminaries}

 \subsection{ Graph Theoretic Notations}
This paper only considers simple, undirected, finite and nontrivial graphs. Let $G=(V,E)$ be a graph. $n$ and $m$ will be used to denote the cardinalities of $V$ and $E$, respectively. $N(v)$ stands for the set of neighbors of a vertex $v$ in $V$. The number of neighbors of a vertex $v\in V$ defines its \emph{degree}, which is represented by the symbol $deg(v)$. The maximum degree of the graph will be denoted by $\Delta$. For a set $U\subseteq V$, the notation $deg_{U}(v)$ is used to represent the number of neighbors that a vertex $v$ has within the subset $U$. Additionally, we use $N_{U}(v)$ to refer to the set of neighbors of vertex $v$ within $U$. Given a set $S\subseteq V$, $G\setminus S$ is defined as the graph induced on $V\setminus S$, that is $G[V\setminus S]$.

A vertex of degree one is known as a \emph{pendant vertex}.  A set $S\subseteq V$ is said to be a \emph{dominating set} if every vertex of $V\setminus S$ is adjacent to some vertex of $S$. A graph 
$G$ is said to be a \emph{complete graph} if any two vertices of $G$ are adjacent. A set $S\subseteq V$ is said to be a clique if the subgraph of $G$ induced on $S$ is a complete graph. A graph is said to be a \emph{cluster graph} if every component of the graph is a clique. For every positive integer $n$,  $[n]$ denotes the set $\{1,2,\ldots,n\}$.

Given a graph $G=(V,E)$ and a function $f:V\to \{0,1,2\}$, $f_H:V(H)\to \{0,1,2\}$ is defined to be the function $f$ restricted on $H$, where $H$ is an induced subgraph of $G$.

For a graph $G = (V, E)$ and a function $f: V \to \{0,1,2\}$; we define $V_i = \{v\in V~\vert~f(v)=i\}$ for $i\in \{0,1,2\}$. The partition  
 $(V_0,V_1,V_2)$ is said to be  \emph{ordered partition} of $V$ induced by $f$. Note that the function $f:V\to \{0,1,2\}$ and the ordered partition $(V_0,V_1,V_2)$ of $V$ have a one-to-one correspondence. So, when the context is clear, we write $f = (V_0, V_1, V_2)$. Given an RDF $f = (V_0,V_1,V_2)$, $(V_1,V_2)$ is said to be \emph{Roman Dominating pair} corresponding to $f$. When the context is clear, we write Roman Dominating pair (omitting the notion of $f$).

\subsection{Problem Definitions}
Before presenting our results, we formalize the problems considered in the paper as follows.\\
\noindent\underline{\textsc{SET-COVER}}
\\
[-13pt]
\begin{enumerate}
  \item[] \textbf{Input}: An universe $U$ and a collection of subsets of $U$, $F=\{S_1,\ldots,S_m\}$ and a non-negative integer $k$. 
  \item[] \textbf{Parameter}: $\vert U\vert$.
  \item[] \textbf{Question}: Does there exists $k$ sets $S_{i_1}, \ldots, S_{i_k}$ in $F$, such that $\bigcup^k_{j=1} S_{i_j}=U$?
\end{enumerate}

\noindent\underline{\textsc{RD-CVD}}
\\
[-13pt]
\begin{enumerate}
  \item[] \textbf{Input}: A graph $G=(V,E)$, a cluster vertex deletion set $S$ and a non-negative integer $\ell$. 
  \item[] \textbf{Parameter}: $\vert S\vert=k$.
  \item[] \textbf{Question}: Does there exists an RDF $f$ on $G$, with weight at most $\ell$?
\end{enumerate}

\noindent\underline{\textsc{RD-VC}}
\\
[-13pt]
\begin{enumerate}
  \item[] \textbf{Input}: A graph $G=(V,E)$, a vertex cover $S$ and a non-negative integer $\ell$. 
  \item[] \textbf{Parameter}: $\vert S\vert=k$.
  \item[] \textbf{Question}: Does there exists an RDF $f$ on $G$, with weight at most $\ell$?
\end{enumerate}

\noindent\underline{\textsc{IRD-CVD}}
\\
[-13pt]
\begin{enumerate}
  \item[] \textbf{Input}: A graph $G=(V,E)$, a cluster vertex deletion set $S$ and a non-negative integer $\ell$. 
  \item[] \textbf{Parameter}: $\vert S\vert=k$.
  \item[] \textbf{Question}: Does there exists an IRDF $f$ on $G$, with weight at most $\ell$?
\end{enumerate}

\noindent\underline{\textsc{$d$-hitting set}}
\\
[-13pt]
\begin{enumerate}
  \item[] \textbf{Input}:  An universe $U$ and a collection of subsets of $U$, $F=\{S_1,\ldots,S_m\}$ such that $\vert S_i\vert\leq d$, for all $i\in [m]$ and a non-negative integer $k$.
  \item[] \textbf{Parameter}: $\vert U\vert$.
  \item[] \textbf{Question}: Does there exists a subset $U'\subseteq U$, such that $U'\cap S_i\neq \emptyset$, for every $i\in [m]$?
\end{enumerate}

\subsection{Parameterized Complexity Notations and Defenitions}\label{sec:para_def}

 \begin{definition}
     (Fixed Parameter Tractability) Let $L \subseteq \Sigma^* \times \mathbb{N}$ be a parameterized language. $L$ is said to be fixed parameter tractable (or FPT) if there exists an algorithm $\mathcal{B}$, a constant $c$ and a computable function $f$ such that for all $x$, for all $k; \mathcal{B}$ on input $(x, k)$ runs in at most $f(k) \cdot|x|^c$ time and outputs $(x, k) \in L$ if and only if $\mathcal{B}([x, k])=1$. We call the algorithm $\mathcal{B}$ as fixed parameter tractable algorithm (or FPT algorithm).
 \end{definition}  

\begin{definition}
    (Parameterized Reduction) Let $P_1, P_2 \in \Sigma^* \times \mathbb{N}$ be two parameterized languages. Suppose there exists an algorithm $\mathcal{B}$ that takes input $(x, k)$ (an instance of $P_1$ ) and constructs an instance $\left(x^{\prime}, k^{\prime}\right)$ of $P_2$ such that the following conditions are satisfied:
    \begin{itemize}
        \item $(x, k)$ is a YES instance if and only if $\left(x^{\prime}, k^{\prime}\right)$ is a YES-Instance.
        \item $k^{\prime} \in f(k)$ for some function depending only on $k$.
        \item Algorithm $\mathcal{B}$ must run in $g(k)|x|^c$ time, where $g(.)$ is a computable function. 
    \end{itemize}
    Then we say that there exists a parameterized reduction from $P_1$ to $P_2$.
\end{definition}

 \noindent\textbf{W-hierarchy}: To capture the parameterized languages being FPT or not, the W-hierarchy is defined as $\mathrm{FPT} \subseteq \mathrm{W}[1] \subseteq \cdots \subseteq \mathrm{XP}$. It is believed that this subset relation is strict~\cite{DBLP:books/sp/CyganFKLMPPS15}. Hence, a parameterized language that is hard for some complexity class above FPT is unlikely to be FPT. Theorem \ref{th:W_hard} gives the use of parameterized reduction. If a parameterized language $L \subseteq \Sigma^* \times \mathbb{N}$ can be solved by an algorithm running in time $\mathcal{O}\left(n^{f(k)}\right)$, then we say $L \in$ XP. In such a situation, we also say that $L$ admits an XP algorithm.

\begin{definition}
    (para-NP-hardness) A parameterized language $L \subseteq \Sigma^* \times \mathbb{N}$ is called para-NP-hard if it is NP-hard for some constant value of the parameter.
\end{definition}

It is believed that a para-NP-hard problem does not admit an XP algorithm as; otherwise it will imply $\mathrm{P}=\mathrm{NP}$ \cite{DBLP:books/sp/CyganFKLMPPS15}.

\begin{theorem}\cite{DBLP:books/sp/CyganFKLMPPS15}\label{th:W_hard}
    Let there be a parameterized reduction from parameterized problem $P_1$ to parameterized problem $P_2$. Then if $P_2$ is FPT, then so is $P_1$. Equivalently, if $P_1$ is $W[i]$-hard for some $i \geq 1$, then so is $P_2$.
\end{theorem}

\begin{definition}
    (Kernelization)  Let $L \subseteq \Sigma^* \times \mathbb{N}$ be a parameterized language. Kernelization is an algorithm that replaces the input instance $(x, k)$ by a reduced instance $\left(x^{\prime}, k^{\prime}\right)$ such that

    \begin{itemize}
        \item $k^{\prime} \leq f(k),\left|x^{\prime}\right| \leq g(k)$ for some functions $f, g$ depending only on $k$.
        \item $(x, k) \in L$ if and only if $\left(x^{\prime}, k^{\prime}\right) \in L$.
    \end{itemize}
    The reduction from $(x, k)$ to $\left(x^{\prime}, k^{\prime}\right)$ must be computable in $p(|x|+k)$ time, where $p(.)$ is a polynomial function. If $g(k)=k^{\mathcal{O}(1)}$ then we say that $L$ admits a polynomial kernel.
\end{definition}

It is well-known that a decidable parameterized problem is FPT if and only if it has a kernel. However, the kernel size could be exponential (or worse) in the parameter. There is a hardness theory for problems having polynomial sized kernel. Towards that, we define the notion of polynomial parameter transformation.

\begin{definition}
    (Polynomial parameter transformation (PPT)) Let $P_1$ and $P_2$ be two parameterized languages. We say that $P_1$ is polynomial parameter reducible to $P_2$ if there exists a polynomial time computable function (or algorithm) $f: \Sigma^* \times \mathbb{N} \rightarrow \Sigma^* \times \mathbb{N}$, a polynomial $p: \mathbb{N} \rightarrow \mathbb{N}$ such that $(x, k) \in P_1$ if and only if $f(x, k) \in P_2$ and $k^{\prime} \leq p(k)$ where $f(x, k)=\left(x^{\prime}, k^{\prime}\right)$. We call $f$ to be a polynomial parameter transformation from $P_1$ to $P_2$.
\end{definition}

The following theorem gives the use of the polynomial parameter transformation for obtaining kernels for one problem from another.

\begin{theorem}\label{th:kernel_hard}\cite{DBLP:journals/tcs/BodlaenderTY11}
    Let $P, Q \subseteq \Sigma^* \times \mathbb{N}$ be two parameterized problems and assume that there exists a PPT from $P$ to $Q$. Furthermore, assume that the classical version of $P$ is NP-hard and $Q$ is in NP. Then, if $Q$ has a polynomial kernel, then $P$ has a polynomial kernel.
\end{theorem}

We use the following conjecture to prove one of our lower bounds.

\begin{conjecture}\cite{DBLP:journals/jcss/ImpagliazzoPZ01}\label{th:SETH}
    (Strong Exponential Time Hypothesis (SETH)) There is no $\epsilon>0$ such that $\forall q \geq 3, q$-CNFSAT can be solved in $(2-\epsilon)^n n^{\mathcal{O}(1)}$ time where $n$ is the number of variables in input formula.
\end{conjecture}

We have the following theorem, which gives an algorithm for SET-COVER parameterized by the size of the universe.

\begin{theorem}\cite{DBLP:series/txtcs/FominK10}
    The SET-COVER problem can be solved in $2^n(m+n)^{\mathcal{O}(1)}$ time where $n$ is the size of the universe and $m$ is the size of the family of subsets of the universe.
\end{theorem} 

 \subsection{Related Works}
From the parameterized complexity point of view, it is surprising that there does not exist much literature on the Roman Domination problem (except \cite{DBLP:conf/walcom/MohanapriyaRS23,DBLP:journals/ijcm/Fernau08}), while the classical dominating set problem is very well studied. Some related literature about the parameterized complexity of the domination problem can be found in \cite{DBLP:journals/jacm/AlberFN04,DBLP:journals/tcs/DowneyF95,DBLP:journals/talg/PhilipRS12}. One recent work about the domination problem parameterized by several structural parameters like distance to cluster, distance to split, can be found in \cite{DBLP:journals/corr/abs-2405-10556}. The techniques we designed in this paper, are adaption of the technique used in \cite{DBLP:journals/corr/abs-2405-10556}, with appropriate modification to fit our problem. 

 In \cite{DBLP:journals/ijcm/Fernau08}, Fernau proved that the Roman Domination parameterized by the solution size is $W[2]$-hard in general graphs, but FPT for planar graphs. He also showed that the same problem parameterized by treewidth is FPT. In \cite{DBLP:conf/walcom/MohanapriyaRS23}, Mohannapriya et al. showed that a more generalized problem, that is $k$-Roman Domination problem parameterized by solution size is $W[1]$-hard, even for split graphs. To the best of our knowledge, no other parameterized complexity results exist for the Roman Domination problem.

 \subsection{Our Results}
The main contribution of the paper is following:
\begin{itemize}

 \item In Section \ref{sec: roman_dom} (resp. Section \ref{sec:ind_rom}), we show that the RD-CVD (resp. IRD-CVD) problem is FPT. 
 
 \item In Section \ref{sec:LB}, we show that the RD-CVD problem cannot be solved in time $2^{\epsilon k} n^{O(1)}$ (where $0<\epsilon<1$) unless SETH fails, neither it admits a polynomial kernel unless NP $\subseteq$ coNP$/$poly. 
 
 \item In Section \ref{sec:conclusion}, we conclude the paper with some future research directions.
\end{itemize}

\section{Variants of Roman Domination parameterized by CVD size}

In this section, we assume that a cluster vertex deletion set $S$ of size $k$ is given with the input graph $G=(V,E)$. If not, the algorithm mentioned in \cite{DBLP:journals/mst/BoralCKP16} can be used, which runs in $1.92^kn^{O(1)}$ time and outputs a cluster vertex deletion set of size at most $k$ or concludes that there does not any cluster vertex deletion set of size at most $k$.

\subsection{Roman Domination}\label{sec: roman_dom}
In this section, we propose an FPT algorithm for the Roman Domination problem when the parameter is  CVD size. 

We consider a CVD set $S$ as a part of the input, where $\vert S\vert=k$. Our algorithm starts with making a guess for $S_1=V_1\cap S$ and $S_2=V_2\cap S$, where $(V_1, V_2)$ is an optimal Roman Dominating pair. At first, a guess of $S_2$ is made from $S$. Then, the vertices of $N[S_2]\cap S$ are deleted from the graph. Then we guess $S_1$ from the remaining $S$ and delete $S_1$ from $S$. 

Note that $S$ is a CVD set, hence $G\setminus S$ is disjoint union of cliques. Let $G\setminus S= C_1\cup C_2\cup\ldots C_q$, where every $C_i$ is a clique and $\vert C_i\vert=\ell_i$, for $i\in [q]$, Note that $q\leq n-k$. After the selection of $S_1$ and $S_2$, every clique belongs to exactly one of the following three types:

Type $0$ ($T_0$) cliques: $C_i$ is a $T_0$ clique if every vertex of $C_i$ is adjacent to at least one vertex of $S_2$.

Type $1$ ($T_1$) cliques: $C_i$ is a $T_1$ clique if exactly one vertex of $C_i$ is not adjacent to any vertex of $S_2$.

Type $2$ ($T_2$) cliques: $C_i$ is a $T_2$ clique if $C_i$  contains at least two vertices which are not adjacent to any vertex of $S_2$.

We define an order $\rho_i$ on the vertices of the clique $C_i$ as follows: if $C_i$ is a $T_0$ or $T_2$ clique, then we order them arbitrarily; if $C_i$ is a $T_1$ clique, then the set $C_i\setminus N[S_2]$ contains exactly one vertex. We make that vertex the first vertex of $\rho_i$ and order the rest of the vertices of $C_i$ arbitrarily. Now we define an order $\rho$ on the vertex set of $G\setminus S$ as follows: $\rho=v_1,v_2,\ldots,v_{\vert G\setminus S\vert}$, where first $\ell_1$ vertices of $\rho$ are vertices of $C_1$ and follows the order $\rho_1$, then the next  $\ell_2$ vertices of $\rho$ are vertices of $C_2$ and follows the order $\rho_2$ and so on. Now, here comes an observation.

\begin{observation}
    Given a $T_2$ clique $C_i$, any Roman Dominating pair $(V_1,V_2)$ extended from $(S_1,S_2)$ has one of the following properties:
    \begin{enumerate}
        \item $C_i\cap V_2\neq \emptyset$.
        \item A Roman Dominating pair $(V'_1,V'_2)$ can be  extended from $(S_1,S_2)$, which has same or less weight than $(V_1,V_2)$ and $V'_2\cap C_i\neq \emptyset$
    \end{enumerate}
\end{observation}

\begin{proof}
    As $C_i$ is a $T_2$ clique, let us assume that $v_1$ and $v_2$ are two vertices of $C_i$ which are not adjacent to $S_2$. If either of $v_1, v_2$ belongs to $V_2$, then we are done. If not, then two cases may arise:
    
    \noindent \textbf{Case 1}: At least one vertex among $v_1,v_2$ has label $0$. Without loss of generality, let $v_1$ has label $0$. Then, one neighbour $v$ of $v_1$ must have label $2$. But $v$ can not belong to $S$ as $V_2\cap S=S_2$ and $v_1$ has no neighbour in $S_2$. Hence, $v$ must belong to $C_i$, which implies that $C_i\cap V_2\neq \emptyset$.

    \noindent \textbf{Case 2}: None of $v_1,v_2$ has label $0$, which implies that both have label $1$. Now we construct a dominating pair $(V'_1,V'_2)$ as follows: $V'_2=V_2\cup\{v_1\}$ and $V'_1=V_1\setminus\{v_1,v_2\}$. Note that $(V'_1,V'_2)$ is a Roman Dominating pair and $(V'_1,V'_2)$ has the same weight as $(V_1,V_2)$ and $V'_2\cap C_i\neq \emptyset$. Hence, the result follows.      
\end{proof}

From the above observation, we can rephrase the remaining problem as follows:

\medskip
\noindent\underline{\textsc{RD-DisjointCluster} problem}
\\
[-13pt]
\begin{enumerate}
  \item[] \textbf{Input}: A graph $G=(V,E)$, a subset $S\subseteq V$ such that every component of $G\setminus S$ is a clique, a $(0,1,2)$-flag vector $f=(f_1,f_2,\ldots,f_q)$ corresponding to the cliques $(C_1,C_2,\ldots,C_q)$ and $\ell\in \mathbb{Z^+}$.
  \item[] \textbf{Parameter}: $\vert S\vert$.
  \item[] \textbf{Question}: Does there exists a subset $T\subseteq G\setminus S$ which satisfies all of the following conditions?
  \begin{enumerate}
      \item For every $C_i$ with $f_i=2$, $T\cap C_i\neq \emptyset$.
      \item $2\vert T\vert + g(T)\leq \ell$, where $g(T)=$ number of cliques with flag $1$, which have empty intersection with $T$.
  \end{enumerate}
\end{enumerate}

For an instance $(G,S,\ell)$ of the RD-CVD problem, with cliques $C_1,\ldots,C_q$ and the guesses of $S_1,S_2$, we build an instance $(\hat{G}, \hat{S}, f,\hat{\ell})$ of the \textsc{RD-DisjointCluster} problem as follows:

\begin{itemize}
    \item $\hat{G}=G\setminus ((N[S_2]\cap S)\cup  S_1$).
    \item $\hat{S}=S\setminus ((N[S_2]\cap S)\cup  S_1)$.
    \item $\hat{\ell}=\ell-2\vert S_2\vert-\vert S_1\vert$.
    \item For all $i\in [q]$, $f_i=j$, if $C_i$ is a $T_j$ clique, $j\in\{0,1,2\}$.
\end{itemize}

\noindent \textbf{Formulation of the problem as a variant of set cover}: We define a variant of the set cover problem. Given an instance of the \textsc{RD-DisjointCluster} problem, we construct an instance of the set cover problem. Let $(G,S,f,\ell)$  be an instance of the \textsc{RD-DisjointCluster} problem and $\rho=(v_1,v_2,\ldots,v_{\vert G\setminus S\vert})$ be an ordering of the vertex set of $G\setminus S$, as defined earlier. We take the universe $U=S$ and $F=\{S_1,S_2,\ldots,S_{\vert G\setminus S\vert}\}$, where $S_i=N(v_i)\cap S$ for every $i\in [\vert G\setminus S\vert]$. Now, we modify the usual SET-COVER problem to suit our problem. The modified SET-COVER problem is defined below:

\medskip
\noindent\underline{\textsc{SET-CoverWithPartition} problem (SCP)}
\\
[-13pt]
\begin{enumerate}
  \item[] \textbf{Input}: Universe $U$, a family of sets $F=\{S_1,S_2,\ldots,S_m\}$, a partition of $\beta=(\beta_1,\ldots,\beta_q)$ of $F$, a $(0,1,2)$ flag vector $f=(f_1,f_2,\ldots,f_q)$ corresponding to each block in the partition $\beta$ and a non-negative integer $\ell$.
  \item[] \textbf{Parameter}: $\vert U\vert$.
  \item[] \textbf{Question}: Does there exists a subset $F'\subseteq F$ which satisfies all of the following conditions?
  \begin{enumerate}
      \item For every $\beta_i$ with $f_i=2$, $F'\cap \beta_i\neq \emptyset$.
      \item $2\vert F'\vert + g(F')\leq \ell$, where $g(A)=$ number of blocks with flag $1$, which have empty intersection with $A$, for $A\subseteq F$.
  \end{enumerate}
\end{enumerate}

Given an instance $(G,S,f,\ell)$ of the \textsc{RD-DisjointCluster} problem, we define an instance $(U,F,\beta,f',\ell')$ of the SCP problem as follows:

\begin{itemize}
    \item $U=S$.
    \item $F=\{S_i=N(v_i)\cap S~\vert~v_i\in G\setminus S\}$.
    \item $\beta_i=\{S_j~\vert~v_j\in C_i\}$, for $i\in [q]$.
    \item $f'=f$.
    \item $\ell'=\ell$.
\end{itemize}

It is not hard to show that the \textsc{RD-DisjointCluster} and SCP are equivalent problems.

\begin{observation}
    $(G,S,f,\ell)$ is a YES instance of the \textsc{RD-DisjointCluster} problem if and only if $(U,F,\beta,f',\ell')$ is a YES instance of the SCP problem.
\end{observation}

Now, we propose an algorithm to solve the SCP problem.

\begin{theorem}
    The \textsc{Set-CoverWithPartition} problem can be solved in $2^{\vert U\vert}O(m\cdot\vert U\vert)$ time.
\end{theorem}

\begin{proof}
    We propose a dynamic programming algorithm to solve the problem. For every $W\subseteq U$, $j\in [m]$ and $b\in \{0,1,2\}$, we define $OPT[W,j,b]:=min_{X}\{2\vert X\vert+g_j(X)\}$, where $X$ satisfies the following properties:
    
    \begin{enumerate}
        \item $X\subseteq \{S_1,\ldots,S_j\}$.
        \item $X$ covers $W$.
        \item Let $\beta_x$ be the block that contains $S_j$. We redefine $f_x=b$, where $f_x$ is the flag associated with $\beta_x$. From every block $\beta_i$ ($i\leq x$) with $f_i=2$, at least one set from $\beta_i$ is in $X$.
    
        \item The function $g_j$ is defined as follows. $g_j(X):=$ number of blocks $\beta_i~(i\leq x)$ with  $f_i=1$, which have empty intersection with $X$. 
    \end{enumerate} 

    Now, coming to the base case, for every $W\subseteq U$, with $W\neq \emptyset$ and $b\in \{0,1,2\}$; $OPT[W,1,b]=2$ if $W\subseteq S_1$, $OPT[W,1,b]=\infty$, otherwise. 
    
    If $W=\emptyset$, $OPT[W,1,b]=b$, for $b\in \{0,1,2\}$. To compute all the values of $OPT[W,j,b]$, we initially set all the remaining values to be $\infty$. We construct the following recursive formulation for $OPT[W,j+1,b]$, for $j\geq 1$:

    \noindent \textbf{Case 1}: $S_{j+1}$ is not the first set of the block $\beta_x$.

    Note that two possibilities appear here. First, we pick $S_{j+1}$ in the solution $X$. Hence, we are left with the problem of covering $W\setminus S_{j+1}$ with some subset of $\{S_1,\ldots,S_j\}$ and since $S_{j+1}$ from the partition $\beta_x$ is already taken in solution, so the flag of $\beta_x$ can be reset to $0$. Hence, in this case $OPT[W,j+1,b]=2+ OPT[W\setminus S_{j+1},j,0]$. 
    
    In the latter case, we do not pick $S_{j+1}$ in $X$; hence nothing is changed except the fact that now we need to cover $W$ with a subset of $\{S_1,\ldots, S_j\}$ and the flag of $\beta_x$ remains unchanged as $b$. Hence, $OPT[W,j+1,b]= OPT[W,j,b]$.

    So, $OPT[W,j+1,b]=\min\{2+ OPT[W\setminus S_{j+1},j,0], OPT[W,j,b]\}$.

    \noindent \textbf{Case 2}: $S_{j+1}$ is the first set of the block $\beta_x$. Here, three scenarios can appear:

    \textbf{Case 2.1}: $b=2$.

    In this case, there is no option but to include $S_{j+1}$ in the solution as $b=2$. Hence, we take $S_{j+1}$ in the solution and shift to the previous block. Now we need to cover $W\setminus S_{j+1}$ with a subset of $\{S_1,\ldots,S_j\}$. Hence, $OPT[W,j+1,b]= 2+OPT[W\setminus S_{j+1},j,f_{x-1}]$.

    \textbf{Case 2.2}: $b=0$.
    
    In this case, there are two choices, to include $S_{j+1}$ in the solution or not. If we include $S_{j+1}$ in the solution, then $OPT[W,j+1,b]= 2+OPT[W\setminus S_{j+1},j,f_{x-1}]$. If we do not, then $OPT[W,j+1,b]= OPT[W,j,f_{x-1}]$. Hence $OPT[W,j+1,b]= \min\{2+OPT[W\setminus S_{j+1},j,f_{x-1}], OPT[W,j,f_{x-1}]\}$

    \textbf{Case 2.3}: $b=1$.

    Similarly, in this case, there are two choices. If $S_{j+1}$ is included in the solution then $OPT[W,j+1,b]= 2+OPT[W\setminus S_{j+1},j,f_{x-1}]$, by similar argument as above. If not, then $S_{j+1}$ has to contribute $1$ in $OPT[W,j+1,b]$, as at least one set from the block $\beta_x$ has to contribute $1$ to $OPT[W,j+1,b]$ and $S_{j+1}$ is the only set left in $\beta_x$ at this moment. So, in this case, $OPT[W,j+1,b]=1+OPT[W,j,f_{x-1}]$. Hence, $OPT[W,j+1,b]=\min\{2+OPT[W\setminus S_{j+1},j,f_{x-1}], 1+OPT[W,j,f_{x-1}]\}$.

    We compute $OPT[W,j,b]$ in the increasing order of size of $W,j,b$. Hence, there are $3\cdot 2^{\vert U\vert}\cdot m$ subproblems. It takes $\vert U\vert$ time to compute set differences (like $W\setminus S_{j+1}$). Hence, the time-complexity of our algorithm is $2^{\vert U\vert}O( m\cdot \vert U\vert)$.
\end{proof}

Hence, the following corollary can be concluded.

\begin{corollary}
    The \textsc{RD-DisjointCluster} problem can be solved in time $2^{\vert S\vert}n^{O(1)}$.
\end{corollary}

\begin{theorem}\label{th:1_main}
    The RD-CVD problem can be solved in time $4^kn^{O(1)}$.
\end{theorem}

\begin{proof}
    Given an instance $(G,S,\ell)$ of the RD-CVD problem and for every guess of $S'_1, S'_2\subseteq S$ (with $\vert S'_1\vert=i_1$ and $\vert S'_2\vert=i_2$), we can construct an instance $(\hat{G},\hat{S},f,\hat{\ell})$ of the \textsc{RD-DisjointCluster} problem, which can be solved in time $2^{k-i_1-i_2}n^{O(1)}$. Hence, total time taken is $\sum_{i_1=1}^{k} ({k\choose i_1}\sum_{i_2=1}^{k-i_1}({k-i_1\choose i_2}2^{k-i_1-i_2}))n^{O(1)}$=$4^kn^{O(1)}$.
\end{proof}

In the next section, using a similar approach, we show that the IRD-CVD problem is also FPT.

\subsection{Independent Roman Domination}\label{sec:ind_rom}

In this section, we propose an FPT algorithm for the Independent Roman Domination problem when the parameter is the CVD size.

Similarly, like the case of Roman Domination, a guess for $S_1=V_1\cap S$ and $S_2=V_2\cap S$ is made, where $(V_1, V_2)$ is an optimal independent Roman Dominating pair. At first, we guess an independent set $S_2$ from $S$ and then delete all the vertices of $N[S_2]$ from the graph, as if our choice of $S_2$ is right, then all the vertices in $N(S_2)\setminus S_2$ should have label $0$. Then, we choose another independent set $S_1$ from the remaining $S$ and delete all the vertices of $S_1$ and $N(S_1)\cap (G\setminus S)$ from the remaining graph. Note that if there exists a clique $C_i\subseteq N[S_1]\cup N[S_2]$  such that $C_i$ has a vertex $v$, that is not adjacent to any vertex of $S_2$, but it is adjacent to some vertex in $S_1$, then our choices of $S_1$, $S_2$ are incorrect, and we do not move further with these choices of $S_1$ and $S_2$.

Note that $S$ is a CVD set, hence $G\setminus S$ is disjoint union of cliques. Let $G\setminus S= C_1\cup C_2\cup\ldots C_q$, where every $C_i$ is a clique and $\vert C_i\vert=\ell_i$, for $i\in [q]$. Note that $q\leq n-k$. Note that, after the selection of $S_1$ and $S_2$ and the deletion process, every clique belongs to exactly one of the following two types:

Type $1$ ($T_1$) cliques: $C_i$ is a $T_1$ clique if $C_i$ has exactly one vertex.

Type $2$ ($T_2$) cliques: $C_i$ is a $T_2$ clique if $C_i$ has at least two vertices.

\begin{observation}
    Given a $T_2$ clique $C_i$, any independent Roman Dominating pair $(V_1,V_2)$ extended from $(S_1,S_2)$ has the following property: $C_i\cap V_2\neq \emptyset$.
\end{observation}
\begin{proof}
    $C_i$ is a $T_2$ clique, hence there exist at least two vertices $v_1, v_2$ in $C_i$. Note that both of them can not have non zero labels; at least one of them must have label $0$. Without loss of generality, let $v_1$ have label $0$, but $v_1$ does not have any neighbor in $S$, which has label $2$, which implies a vertex in $C_i$ must have label $2$. Hence, the result follows. 
\end{proof}

Hence the remaining problem can be rephrased as follows:

\medskip
\noindent\underline{\textsc{IRD-DisjointCluster} problem}
\\
[-13pt]
\begin{enumerate}
  \item[] \textbf{Input}: A graph $G=(V,E)$, a subset $S\subseteq V$ such that every component of $G\setminus S$ is a clique, a $(1,2)$-flag vector $f=(f_1,f_2,\ldots,f_q)$ corresponding to the cliques $(C_1,C_2,\ldots,C_q)$ and $\ell\in \mathbb{Z^+}$.
  \item[] \textbf{Parameter}: $\vert S\vert$.
  \item[] \textbf{Question}: Does there exists a subset $T\subseteq G\setminus S$ which satisfies all of the following conditions?
  \begin{enumerate}
      \item For every $C_i$ with $f_i=2$, $\vert T\cap C_i\vert =1$.
      \item $2\vert T\vert + g(T)\leq \ell$, where $g(T)=$ number of cliques with flag $1$, which have empty intersection with $T$.
  \end{enumerate}
\end{enumerate}

For an instance $(G,S,\ell)$ of the IRD-CVD problem, with cliques $C_1,\ldots,C_q$ and the guesses of $S_1,S_2$, we build an instance $(\hat{G}, \hat{S}, f,\hat{\ell})$ of the \textsc{IRD-DisjointCluster} problem as follows:

\begin{itemize}
    \item $\hat{G}=G\setminus (S_1\cup N[S_2]\cup  (N(S_1)\cap (G\setminus S$)).
    \item $\hat{S}=S\setminus (N[S_2]\cup S_1)$.
    \item $\hat{\ell}=\ell-2\vert S_2\vert-\vert S_1\vert$.
    \item $f_i=j$, if $C_i$ is a $T_j$ clique, $i\in [q]$ and $j\in \{1,2\}$. 
\end{itemize}

\noindent \textbf{Formulation of the problem as a variant of set cover}: We define a variant of the set cover problem similarly to the Roman Domination problem. Given an instance of the \textsc{IRD-DisjointCluster} problem, we construct an instance of the set cover problem. Let $(G,S,\ell,f)$  be an instance of the \textsc{IRD-DisjointCluster} problem and $\rho$ be any arbitrary order of the vertex set $G\setminus S$. We take the universe $U=S$ and $F=\{S_1,S_2,\ldots,S_{\vert G\setminus S\vert}\}$, where $S_i=N(v_i)\cap S$ for every $i\in [\vert G\setminus S\vert]$. Now, we modify the usual set cover problem to suit our problem. The modified set cover problem is defined below:

\medskip
\noindent\underline{\textsc{Independent-Set-CoverWithPartition} problem (ISCP)}
\\
[-13pt]
\begin{enumerate}
  \item[] \textbf{Input}: Universe $U$, a family of sets $F=\{S_1,S_2,\ldots,S_m\}$, a partition of $\beta=(\beta_1,\ldots,\beta_q)$ of $F$, a $(1,2)$ flag vector $f=(f_1,f_2,\ldots,f_q)$ corresponding to each block in the partition $\beta$ and a non-negative integer $\ell$.
  \item[] \textbf{Parameter}: $\vert U\vert$.
  \item[] \textbf{Question}: Does there exists a subset $F'\subseteq F$ which satisfies all of the following conditions?
  \begin{enumerate}
      \item For every $\beta_i$ with $f_i=2$, $\vert F'\cap \beta_i\vert=1$.
      \item $2\vert F'\vert + g(F')\leq \ell$, where $g(A)=$ number of blocks with flag $1$, which have empty intersection with $A$; for $A\subseteq F$.
  \end{enumerate}
\end{enumerate}

Given an instance $(G,S,f,\ell)$ of the \textsc{IRD-DisjointCluster} problem, we define an instance $(U,F,\beta,f',\ell')$ of the ISCP problem as follows:

\begin{itemize}
    \item $U=S$.
    \item $F=\{S_i=N(v_i)\cap S~\vert~v_i\in G\setminus S\}$.
    \item $\beta_i=\{S_j~\vert~v_j\in C_i\}$, for $i\in [q]$.
    \item $f'=f$.
    \item $\ell'=\ell$.
\end{itemize}

\begin{observation}
    $(G,S,f,\ell)$ is a YES instance of the \textsc{IRD-DisjointCluster} problem if and only if $(U,F,\beta,f',\ell')$ is a YES instance of the ISCP problem.
\end{observation}

Next we propose an algorithm to solve the ISCP problem. 

\begin{theorem}
    The \textsc{Independent-Set-CoverWithPartition} problem can be solved in $2^{\vert U\vert}O( m\cdot\vert U\vert)$ time.
\end{theorem}

\begin{proof}
    We propose a dynamic programming algorithm to solve the problem, similar to the algorithm used in the previous section, with slight modifications. For every $W\subseteq U$, $j\in [m]$ and $b\in \{1,2\}$, we define $OPT[W,j,b]:=min_{X}\{2\vert X\vert+g_j(X)\}$, where $X$ satisfies the following properties:
    
    \begin{enumerate}
        \item $X\subseteq \{S_1,\ldots,S_j\}$.
        \item $X$ covers $W$.
        \item Let $\beta_x$ be the block that contains $S_j$. We redefine $f_x=b$, where $f_x$ is the flag associated with $\beta_x$. From every block $\beta_i$ ($i\leq x$) with $f_i=2$, exactly one set from $\beta_i$ is in $X$.
    
        \item $g_j(X):=$ number of blocks $\beta_i~(i\leq x)$ with  $f_i=1$, which have empty intersection with $X$.
    \end{enumerate}

    The base cases are defined as follows:

    If $S_j\in \beta_1$ and $W\neq \emptyset$, then $OPT[W,j,b]=2$ if $W\subseteq S_i$, for some $i\leq j$ and $OPT[W,j,b]=\infty$ otherwise.

    If $S_j\in \beta_1$ and $W=\emptyset$, then $OPT[W,j,b]=b$.  To compute all the values of $OPT[W,j,b]$, we initially set all the remaining values to be $\infty$. We construct the following recursive formulation for $OPT[W,j+1,b]$, (where $S_{j+1}\notin \beta_1$):

    \noindent \textbf{Case 1}: $S_{j+1}$ is not the first set of the block $\beta_x$.

    Then, two choices appear. The first one is to include $S_{j+1}$ in the solution, then we are left with the problem of covering $W\setminus S_{j+1}$ by a subset of $\{S_1,\ldots, S_k\}$, where $S_k$ is the last set in the block $\beta_{x-1}$. Hence for this choice, $OPT[W,j+1,b]=2+OPT[W\setminus S_{j+1},k,f_{x-1}]$.

    If $S_{j+1}$ is not  included in the solution, then $OPT[W,j+1,b]=OPT[W,j,b]$. Hence $OPT[W,j+1,b]=\min\{2+OPT[W\setminus S_{j+1},k,f_{x-1}], OPT[W,j,b]\}$.

    \noindent \textbf{Case 2}: $S_{j+1}$ is the first set of the block $\beta_x$.

    \textbf{Case 2.1}: $b=2$.\\
    In this case, there is no option but to include $S_{j+1}$ in the solution as $b=2$. Hence, we shift to the previous block, and now we need to cover $W\setminus S_{j+1}$ with a subset of $\{S_1,\ldots,S_j\}$. Hence, $OPT[W,j+1,b]= 2+OPT[W\setminus S_{j+1},j,f_{x-1}]$.

    \textbf{Case 2.2}: $b=1$.\\
    In this case, there are two choices. If $S_{j+1}$ is included in the solution then $OPT[W,j+1,b]= 2+OPT[W\setminus S_{j+1},j,f_{x-1}]$, by similar argument as above. If not, then $S_{j+1}$ has to contribute $1$ in $OPT[W,j+1,b]$, as at least one set from the block $\beta_x$ has to contribute $1$ to $OPT[W,j+1,b]$ and $S_{j+1}$ is the only set left in $\beta_x$ at this moment. So, in this case, $OPT[W,j+1,b]=1+OPT[W,j,f_{x-1}]$. Hence, $OPT[W,j+1,b]=\min\{2+OPT[W\setminus S_{j+1},j,f_{x-1}], 1+OPT[W,j,f_{x-1}]\}$.

     We compute $OPT[W,j,b]$ in the increasing order of size of $W,j,b$. Hence, there are $2\cdot 2^{\vert U\vert}\cdot m$ subproblems. Each subproblem takes $\vert U\vert$ time to compute set differences (like $W\setminus S_{j+1}$). Hence, the time-complexity of our algorithm is $2^{\vert U\vert}O( m\cdot \vert U\vert)$.
\end{proof}

Hence, the following corollary can be concluded.

\begin{corollary}
    The \textsc{IRD-DisjointCluster} problem can be solved in time $2^{\vert S\vert}n^{O(1)}$.
\end{corollary}

\begin{theorem}
    The IRD-CVD problem can be solved in time $4^kn^{O(1)}$.
\end{theorem}

\begin{proof}
   The proof is analogous to the proof of Theorem \ref{th:1_main}.
\end{proof}

\section{Lower Bounds}\label{sec:LB}

In this section, we propose a lower bound on the time-complexity of the RD-CVD problem. We also show that the RD-CVD (resp. RD-VC) problem does not admit a polynomial kernel (recall that RD-VC is the Roman Domination problem parameterized by vertex cover number).

First, we provide the lower bound for the RD-CVD problem. Below, we state a necessary result from the existing literature (refer to Theorem 1.1 in \cite{10.1145/2925416}).

\begin{theorem}\label{th:LB1}\cite{10.1145/2925416}
    The following statement is equivalent to SETH: 
    For every $\epsilon<1$, there exists $d\in \mathbb{Z^+}$, such that the $d$-HITTING SET problem for set systems over $[n]$ can not be solved in time $O(2^{\epsilon n})$.
\end{theorem}

Now, we show a reduction from the $d$-HITTING SET problem to the RD-CVD problem to show a similar lower bound like Theorem \ref{th:LB1} for the RD-CVD problem.

\begin{theorem}\label{th:LB2}
    There exists a polynomial time algorithm that takes an instance $(U,F,t)$ of the $d$-HITTING SET problem and outputs an instance $(G,2t)$ of the RD-CVD problem (and the RD-VC problem), where $G$ has a cluster vertex deletion set (and vertex cover) of size $\vert U\vert$; and $(U,F)$ has a $d$-hitting set of size at most $t$ if and only if $G$ has a Roman Dominating function of size at most $2t$.
\end{theorem}

\begin{proof}
    Consider a $d$-HITTING SET instance $(U,F,t)$, where $U=\{u_1,\ldots,u_n\}$ and $F=\{S_1,\ldots,S_m\}$. We construct a graph $G=(V,E)$ as follows:
    \begin{itemize}
        \item $V=U\cup F_1\cup F_2$, where the vertices of $U$ correspond to elements of the universe $U$ and vertices in $F_1$ and $F_2$ correspond to sets in $F$. $U=\{u_1,\ldots,u_n\}$ and $F_i=\{s^i_j~\vert~S_j\in F\}$, for $i\in \{1,2\}$.
        \item $E=E_1\cup E_2\cup E_3$, where $E_1=\{u_iu_j~\vert~i,j\in [n],~i\neq j\}$, $E_2=\{u_is^1_j~\vert~u_i\in S_j\}$ and $E_2=\{u_is^2_j~\vert~u_i\in S_j\}$.
    \end{itemize}

    Note that $G$ is a split graph, where $U$ is a clique, and $F_1\cup F_2$ is the independent set. $U$ is a cluster vertex deletion set and a vertex cover of $G$.

    Let $(U,F,t)$ be a YES instance, which implies that there exists $S\subseteq U$, such that $S\cap S_i\neq \emptyset$, for every $i\in [m]$, and $\vert S\vert\leq t$. Now, we define a function on $V$ as follows: $f:V\to \{0,1,2\}$, where $f(v)=2$ for $v\in S$ and $f(v)=0$, otherwise. Now, since $S\subseteq U$, $f(v)=0$ for every $v\in F_1\cup F_2$. For every $s^i_j\in F_i$, there must exist $u_k\in S$ which is adjacent to $s^i_j$, since $S$ must has an element $u_k$ which hits the set $S_j$. For every $v\in U\setminus S$, with $f(v)=0$ is adjacent to some vertex with $f$-value $2$, as $G[U]$ is a clique and $S$ is nonempty. Hence, we can conclude that $f$ is an RDF with weight at most $2\vert S\vert\leq 2t$.

    Conversely, let $f$ be an RDF with $w(f)\leq 2t$. First, we prove the following claim.

    \begin{claim}\label{Claim_LB}
        There exists an RDF $g$ on $V$ with weight at most $w(f)$, which satisfies the property: $g(v)=0$, for all $v\in F_1\cup F_2$.
    \end{claim}

    \begin{claimproof}
        Let there exist $v\in F_1\cup F_2$, such that $f(v)> 0$. Two cases may appear:
        
        \noindent \textbf{Case 1}: $f(v)=1$

        Without loss of generality, let $v\in F_1$, hence $v=s^1_j$ for some $j\in [m]$. If $s^1_j$ has a neighbour $v'$ which has non zero label, then define a function $f':V\to \{0,1,2\}$ as follows: $f'(v')=2$ and $f'(s^1_j)=0$ and $f'(u)=f(u)$ for every other $u\in V$. Note that $f'$ is an RDF and $w(f')\leq w(f)$.
        
        Now, if every neighbour of $s^1_j$ has label $0$ under $f$, then observe that $f(s^2_j)>0$. Hence, we define a function $f':V\to \{0,1,2\}$ as follows: $f'(v')=2$, $f'(s^1_j)=f'(s^2_j)=0$, where $v'$ is a neigbour of $s^1_j$ and $f'(u)=f(u)$ for every other $u\in V$. Note that $f'$ is an RDF and $w(f')\leq w(f)$.

        \noindent \textbf{Case 2}: $f(v)=2$

        In this case, we define a function $f':V\to \{0,1,2\}$ as follows: $f'(v')=2$, $f'(v)=0$, where $v'$ is a neigbour of $v$ and $f'(u)=f(u)$ for every other $u\in V$. Note that $f'$ is an RDF and $w(f')\leq w(f)$.
        
        Hence, in both cases, we can define another RDF $f'$ of the same or less weight than $f$, such that $\vert (V^{f'}_1\cup V^{f'}_2)\cap (F_1\cup F_2)\vert< \vert (V^{f}_1\cup V^{f}_2)\cap (F_1\cup F_2)\vert$. Hence, applying this technique iteratively, we get an RDF $g$, with $w(g)\leq w(f)$ and $\vert (V^{g}_1\cup V^{g}_2)\cap (F_1\cup F_2)\vert=0$. Hence, the claim follows. 
    \end{claimproof}

    Hence, by Claim \ref{Claim_LB}, we consider an RDF $f$ on $V$, such that $f(v)=0$, for every $v\in F_1\cup F_2$ and $w(f)\leq 2t$. We define $S=\{v\in U~\vert~f(v)=2\}$. Note that $\vert S\vert\leq \frac{2t}{2}=t$ and every vertex of $F_1\cup F_2$ is adjacent to some vertex of $S$. This implies that $S$ is a hitting set of $(U, F)$ of cardinality at most $t$. Hence, the theorem follows.
\end{proof}

Hence, by Theorem \ref{th:LB1} and \ref{th:LB2}, we can prove the following theorem.

\begin{theorem}\label{th:LB3}
    The RD-CVD (and RD-VC) problem can not be solved in time $2^{\epsilon k}n^{O(1)}$, for any $0<\epsilon<1$, unless SETH fails. \end{theorem}

\begin{proof}
    Let the RD-CVD (or RD-VC) problem be solved in time $2^{\epsilon k}n^{O(1)}$. Then by Theorem \ref{th:LB2}, we can solve the $d$-HITTING SET problem with $\vert U\vert=k$ in $2^{\epsilon k}n^{O(1)}$ time. Hence, by Theorem \ref{th:LB1}, this contradicts the SETH. Hence, the RD-CVD (and RD-VC) problem can not be solved in time $2^{\epsilon k}n^{O(1)}$, unless the SETH fails.
\end{proof}

We state another theorem to show that the RD-CVD (and RD-VC) problem is unlikely to admit a polynomial kernel.

\begin{theorem}\cite{10.1145/2650261}\label{th:LB4}
    The $d$-HITTING SET problem parameterized by the universe size does not admit any polynomial kernel unless NP $\subseteq$ coNP$/$poly.
\end{theorem}

Hence combining the above discussion with Theorem \ref{th:LB4}, the  following theorem can be concluded.

\begin{theorem}\label{th:LB5}
    The RD-CVD (and RD-VC) problem does not admit a polynomial kernel unless NP $\subseteq$ coNP$/$poly. 
\end{theorem}

\begin{proof}
    By Theorem \ref{th:LB4}, the $d$-HITTING SET problem parameterized by universe size does not admit a polynomial kernel unless NP $\subseteq$ coNP$/$poly. Since the reduction provided in Theorem \ref{th:LB2} is a polynomial parameter transformation (PPT), by Theorem \ref{th:kernel_hard}, the RD-CVD and RD-VC do not admit a polynomial kernel unless NP $\subseteq$ coNP$/$poly.
\end{proof}

\section{Conclusion}\label{sec:conclusion}
\vspace*{-.3cm}

In this work, we have extended the study on the parameterized complexity of the Roman Domination problem and one of its variants. There are other interesting structural parameters, such as neighborhood diversity and cliquewidth, for which it would be interesting to determine whether the problem parameterized by these parameters is fixed parameter tractable (FPT).

Another promising research direction is to develop an algorithm to solve the RD-CVD problem with better time-complexity than $4^{k}n^{O(1)}$. Given that the lower bound on time-complexity mentioned in Theorem \ref{th:LB3}, it might be possible to achieve an improved algorithm.

\bibliography{lipics-v2021-sample-article}

\begin{thebibliography}{10}

\bibitem{DBLP:journals/jacm/AlberFN04}
Jochen Alber, Michael~R. Fellows, and Rolf Niedermeier.
\newblock Polynomial-time data reduction for dominating set.
\newblock {\em J. {ACM}}, 51(3):363--384, 2004.

\bibitem{DBLP:journals/tcs/BodlaenderTY11}
Hans~L. Bodlaender, St{\'{e}}phan Thomass{\'{e}}, and Anders Yeo.
\newblock Kernel bounds for disjoint cycles and disjoint paths.
\newblock {\em Theor. Comput. Sci.}, 412(35):4570--4578, 2011.

\bibitem{DBLP:journals/mst/BoralCKP16}
Anudhyan Boral, Marek Cygan, Tomasz Kociumaka, and Marcin Pilipczuk.
\newblock A fast branching algorithm for cluster vertex deletion.
\newblock {\em Theory Comput. Syst.}, 58(2):357--376, 2016.

\bibitem{10.1145/2925416}
Marek Cygan, Holger Dell, Daniel Lokshtanov, D\'{a}niel Marx, Jesper Nederlof, Yoshio Okamoto, Ramamohan Paturi, Saket Saurabh, and Magnus Wahlstr\"{o}m.
\newblock On problems as hard as cnf-sat.
\newblock {\em ACM Trans. Algorithms}, 12, may 2016.

\bibitem{DBLP:books/sp/CyganFKLMPPS15}
Marek Cygan, Fedor~V. Fomin, Lukasz Kowalik, Daniel Lokshtanov, D{\'{a}}niel Marx, Marcin Pilipczuk, Michal Pilipczuk, and Saket Saurabh.
\newblock {\em Parameterized Algorithms}.
\newblock Springer, 2015.

\bibitem{10.1145/2650261}
Michael Dom, Daniel Lokshtanov, and Saket Saurabh.
\newblock Kernelization lower bounds through colors and ids.
\newblock {\em ACM Trans. Algorithms}, 11, 2014.

\bibitem{DBLP:journals/tcs/DowneyF95}
Rodney~G. Downey and Michael~R. Fellows.
\newblock Fixed-parameter tractability and completeness {II:} on completeness for {W[1]}.
\newblock {\em Theor. Comput. Sci.}, 141(1{\&}2):109--131, 1995.

\bibitem{DBLP:journals/ijcm/Fernau08}
Henning Fernau.
\newblock {ROMAN} {DOMINATION:} a parameterized perspective.
\newblock {\em Int. J. Comput. Math.}, 85(1):25--38, 2008.

\bibitem{DBLP:series/txtcs/FominK10}
Fedor~V. Fomin and Dieter Kratsch.
\newblock {\em Exact Exponential Algorithms}.
\newblock Texts in Theoretical Computer Science. An {EATCS} Series. Springer, 2010.

\bibitem{DBLP:journals/corr/abs-2405-10556}
Dishant Goyal, Ashwin Jacob, Kaushtubh Kumar, Diptapriyo Majumdar, and Venkatesh Raman.
\newblock Parameterized complexity of dominating set variants in almost cluster and split graphs.
\newblock {\em CoRR}, abs/2405.10556, 2024.

\bibitem{DBLP:journals/jcss/ImpagliazzoPZ01}
Russell Impagliazzo, Ramamohan Paturi, and Francis Zane.
\newblock Which problems have strongly exponential complexity?
\newblock {\em J. Comput. Syst. Sci.}, 63(4):512--530, 2001.

\bibitem{DBLP:conf/walcom/MohanapriyaRS23}
A.~Mohanapriya, P.~Renjith, and N.~Sadagopan.
\newblock Roman k-domination: Hardness, approximation and parameterized results.
\newblock In Chun{-}Cheng Lin, Bertrand M.~T. Lin, and Giuseppe Liotta, editors, {\em {WALCOM:} Algorithms and Computation - 17th International Conference and Workshops, {WALCOM} 2023, Hsinchu, Taiwan, March 22-24, 2023, Proceedings}, volume 13973 of {\em Lecture Notes in Computer Science}, pages 343--355. Springer, 2023.

\bibitem{DBLP:journals/talg/PhilipRS12}
Geevarghese Philip, Venkatesh Raman, and Somnath Sikdar.
\newblock Polynomial kernels for dominating set in graphs of bounded degeneracy and beyond.
\newblock {\em {ACM} Trans. Algorithms}, 9(1):11:1--11:23, 2012.

\bibitem{stewart1999defend}
Ian Stewart.
\newblock Defend the roman empire!
\newblock {\em Scientific American}, 281:136--138, 1999.

\end{thebibliography}

\end{document}